\documentclass[a4paper,10pt]{article}

\usepackage{amssymb}
\usepackage{amsmath}
\usepackage{amsthm}

\newcommand{\spa}{{\rm span}}
\newcommand{\head}{{\rm head}}
\newcommand{\tail}{{\rm tail}}
\newcommand{\rank}{{\rm rank}}
\newcommand{\dem}{{\rm dem}}

\newcommand{\maxflow}{{\rm maxflow}}

\newtheorem{theorem}{Theorem}[]

\newtheorem{lemma}[theorem]{Lemma}

\title{The Limitation of Random Network Coding}
\author{Yuan Li}

\begin{document}

\maketitle

\begin{abstract}
It is already known that in multicast (single source, multiple sinks) network, random linear network coding can achieve the maximum flow upper bound. In this paper,
we investigate how random linear network coding behaves in general multi-source multi-sink case, where each sink has different demands, and we characterize all achievable rate of
random linear network coding by a simple maximum flow condition.
\end{abstract}

\section{Introduction}
In an information transmission network, allowing coding operation at intermediate nodes will increase the capacity of the network than
simply relaying the packets. In multicast (single-source, multiple sinks) scenario, Ahlswede, Cai, Li, and Yeung proved that the maximum flow upper bound can be achieved by network coding in their seminal paper \cite{ACLY00}. In 2003, Li, Yeung, and Cai proved that \cite{LYC03} linear
network coding, i.e., only linear encoding and decoding is allowed, is sufficient to achieve the maximum flow upper bound.
Later in 2003, Koetter and M$\acute{\text{e}}$dard formulated and dervied Li et. al.'s results using algebraic methods \cite{KM03}.
In 2006, Ho et. al. showed that, in fact, performing random linear network coding will achieve the upper bound, when the underlying coding field is
large enough. Due to its simplicity, random linear network coding turns out to be a practial solution.

In contrast to the multicast scenario, the general case (multi-source and multi-sink with arbitrary demands) is not well understood. In \cite{KM03}, Koetter and M$\acute{\text{e}}$dard reduces the existence of linear network coding solution to the exsitence of a point in certain algebraic variety, which, in general, is NP-complete. In 2005, Dougherty, Freiling, and Zeger showed that the linear network coding is not sufficient in the general case \cite{DFZ05}.
In \cite{YYZ07}, Yan, Yeung, and Zhang characterized the capacity region for multi-source multi-sink network coding. However, the region is difficult
to compute. In fact, even approximating the capacity or linear capacity of network coding within any constant was proven to be hard \cite{LS11}.

In this paper, we investigate how random linear network coding behaves in multi-source multi-sink network. It turns out that it will
work in certain occasions, which can be easily characterized by a simple maxflow condition. And there is also a dichotomy of random network coding in the general case: it will work with probability $\to 1$ or fails with probability $\to 1$ when the size of the encoding field tends to infinity.

\section{Notations}

Let's consider a multi-source multi-sink acyclic network, which consists of a directed acyclic graph $G=(V,E)$, sources $S=\{s_1, s_2,
\ldots, s_m\} \subseteq V$, sinks $T=\{t_1, t_2,\ldots, t_n\}$, for each $t \in T$, demand $\dem(t) \subseteq [m]$, which indicates that
$t$ need to receive all symbols from source $\{ s_i : i \in \dem(t)\}$. Rate $(r_1, \ldots, r_m) \in \mathbb{N}^m$ means, for each $s_i$
sends $r_i$ symbols at a time over some underlying finite field $\mathbb{F}$.

When rate $(r_1, \ldots, r_m)$ is fixed, for convenience of description, let's add $m$ extra vertices $s_1^*, \ldots, s_m^*$, and also add $r_i$ edges from $s^*_i$
to $s_i$. And denote by $S^* = \{s^*_1, \ldots, s^*_m\}$, $r = r_1 + \ldots + r_m$.

If $U_1$ is a subspace $\mathbb{F}^r$, and $U_2$ is a complement of $U_1$, then every vector $\alpha \in \mathbb{F}^r$ can written
uniquely as a sum of a vector in $U_1$ and $U_2$, which is denoted by
$$\alpha = \alpha|_{U_1} + \alpha|_{U_2},$$
where $\alpha|_{U_1} \in U_1$ and $\alpha|_{U_2} \in U_2$. In abuse of notation, given a basis of $\mathbb{F}^r$ including $u$, we denote by $\alpha|_u$ the coefficient of $u$ when expressing $\alpha$ in this basis.

A linear network coding $\psi : E \mapsto \mathbb{F}^r$ is recursively defined as follows
$$
\psi(e) = \begin{cases}
b_{i, j}, & \text{$e$ is the $j^{\text{th}}$ edge from $s^*_i$ to $s_i$},\\
\sum_{\tail(e_i) = \head(e)} c_i \psi(e_i), & \text{otherwise},\\
\end{cases}
$$
where $\{ b_{i, j}$, $1 \le i \le m$, $1\le j \le s_i \}$, is the standard orthogonal basis of $\mathbb{F}^r$, and $c_i \in \mathbb{F}$ are coefficients.
When performing a random linear network coding, coefficients $c_i \in \mathbb{F}$ are chosen uniformly at random. If we take $\dem(i) = [m]$ for every source
$i$, we obtain the multicast network coding theorem \cite{ACLY00}, \cite{LYC03}, \cite{HMKK06}.

Let $\maxflow(s, t)$ denote the maximum number of edge-disjoint paths from $s$ to $t$, and $\maxflow(S,t)$ denotes the maximum number of edge-disjoint paths from some $s \in S$ to $t$. By the maxflow-mincut theorem, we know in both cases, the value equals the minimum size of the
$s$-$t$ ($S$-$t$) cut.
\section{Main result}

\begin{lemma}
\label{lem:fun} Assume $\psi$ is a linear network coding. Let $S_1 \subseteq S$, and $U_1$ be the subspace spanned by $\{b_{i, j} : s_i \in S_1, 1 \le j \le r_i\}$. If
$(V_S, V_T)$ is a $S_1$-$t$ cut, then
$$
\psi(e)|_{U_1} \in \spa\{\psi(e_i)|_{U_1} : e_i \in (V_S, V_T) \}
$$
for every edge $e$ with $\head(e) = t$.
\end{lemma}
\begin{proof}
Let $e$ by any edge with $\head(e) = t$. By the definition of linear network coding, we have
\begin{eqnarray*}
\psi(e) & \in & \spa\{\psi(e_1) : \head(e_1) = \tail(e)\} \\
& \in & \spa\{\psi(e_2) : (\exists e_1)(\head(e_1) = \tail(e) \text{ and } \head(e_2) = \tail(e_1))\} \\
& \cdots & \cdots \\
& \in & \spa\{\psi(e_l) : \head(e_l) \in S\setminus S_1, \text{ or } (\exists e_l \in E)(e_l\rightarrow e \text{ and } e_l \in (V_S, V_T))\}.
\end{eqnarray*}
Note that $\psi(e_l)|_{U_1} = 0$ for all $\tail(e_l) \in S\setminus S_1$, we claim
$$
\psi(e)|_{U_1} \in \spa\{\psi(e_l)|_{U_1} : e_l \in (V_S, V_T)\},
$$
which completes our proof.
\end{proof}

The next lemma explains why random linear network coding can achieve maximum flow bound in multicast ($|S|= 1$) network.

\begin{lemma}
\label{lem:rnd} Assume $p_1 = (e_{1,1}, \ldots e_{1, l_1}), \ldots, p_t=(e_{t, 1}, \ldots, e_{t, l_t})$ are $t$ edge-disjoint paths. In random linear network coding, if $\psi(e_{1,1}), \psi(e_{2,1}), \ldots, \psi(e_{m,1})$ are linearly independent, then $\psi(e_{1, l_1}), \psi(e_{2, l_2}), \ldots, \psi(e_{t, l_t})$ are linearly independent with probability $\to 1$ when $|\mathbb{F}| \to +\infty$.
\end{lemma}
\begin{proof}
Since $G$ is acyclic, there is a topological order to add the edges one by one such that when $e_i$ is added, all edges $e_j$ with $\head(e_j) = \tail(e_i)$ are already added. Let's add edges in this order to perform random linear network coding. Assume edges $e_{1, i_1} \in p_1, \ldots, e_{t, i_t} \in p_t$ are added, and assume w.l.o.g the next edge to add is $e_{1, i_1+1}$, it suffices to
show
\begin{eqnarray*}
&\Pr\{\psi(e_{1, i_1+1}), \psi(e_{2, i_2}), \ldots, \psi(e_{t, i_t}) \text{ are linearly independent} \\
&\quad \mid \psi(e_{1, i_1}), \ldots, \psi(e_{t, i_t}) \text{ are linearly independent }\} \ge 1-\frac{1}{|\mathbb{F}|},
\end{eqnarray*}
which will imply $\psi(e_{1, l_1}), \psi(e_{2, l_2}), \ldots, \psi(e_{t, l_t})$ are linearly independent with probability $\ge (1-\frac{1}{|\mathbb{F}|})^{l_1 + \ldots + l_m}$.

By the condition that $\psi(e_{1, i_1}), \ldots, \psi(e_{t, i_t})$ are linearly independent, let's extend them to a basis of $\mathbb{F}^r$, that $u_1 = \psi(e_{1, i_1}), \ldots, u_t = \psi(e_{t, i_t}), u_{t+1}, \ldots, u_r \in \mathbb{F}^r$. Then
\begin{eqnarray*}
& &\Pr\{\psi(e_{1, i_1+1}), \psi(e_{2, i_2}), \ldots, \psi(e_{t, i_t}) \text{ are linearly independent} \mid \cdots\}, \\
& = & \Pr\{\psi(e_{1, i_1+1}) \not\in \spa(\psi(e_{2, i_2}), \ldots, \psi(e_{t, i_t})) \mid \cdots \} \\
& \ge & \Pr\{1 + \sum_{e : \head(e_i) = \tail(e_{1, i_1+1})} c_i \cdot \psi(e_i)|_{u_1} \not= 0 \mid \cdots \} \\
& = & \begin{cases}
         1, & \text{ if $\psi(e_i)|_{u_1} = 0$ for every $e_i$} \\
         1 - 1/|\mathbb{F}|, & \text{otherwise}. \\
      \end{cases}
\end{eqnarray*}
\end{proof}

Combining Lemma \ref{lem:fun} and Lemma \ref{lem:rnd}, we conclude that a random linear network coding can achieve the maximum flow upper bound in multicast network.

Also, from the above lemma by keeping track of $|T|$ collections of paths and taking a union bound, we claim that as long as $|\mathbb{F}| > |T|$, the probability that every sink can successfully decode is nonzero, which implies that there exists a linear network coding solution achieving the maxflow upper bound, which
is first proved in \cite{SET03}. In \cite{FS06}, Feder, Ron and Tavory proved a lower bound of size $\sqrt{2N}(1-o(1))$ by information theory arguments.

The next theorem is our main result, which characterize all achievable rate for multi-source multi-sink network that random linear network coding will work. And it reveals a dichotomy in random network coding: for a given rate, the random linear network
coding either succeed with probability $\to 1$, or with probability $\to 0$ when the size of the coding field goes to infinity.

\begin{theorem}
\label{thm:main}
Rate $(r_1, \ldots, r_m) \in \mathbb{N}^m$ is achievable with probability $\to 1$ when $|\mathbb{F}| \to +\infty$ by random linear network coding if and only if,
for every $t \in T$,
\begin{equation}
\label{equ:cha}
\maxflow(S^* \setminus \bigcup_{i \in \dem(t)}\{s^*_i\}, t) + \sum_{i \in \dem(t)} r_i = \maxflow(S^*,t)
\end{equation}

Morevoer, if the above condition is not satisfied, a random linear network coding will succeed with probability $\to 0$ when $|\mathbb{F}| \to +\infty$.
\end{theorem}
\begin{proof}
For the ``if'' part, assume \eqref{equ:cha} is satisfied, we need to show a random linear network coding can achieve rate $(r_1, \ldots, r_m)$ with probability $\to 1$ when $|\mathbb{F}|$ goes to infinity.

Fix any sink $t \in T$, it's enough to prove with probability $\to 1$ ($|\mathbb{F}| \to +\infty$), sink $t$ can decode all symbols from sources $\dem(t)$, i.e.,
$$\{ b_{i,j} : i \in \dem(t), 1 \le j \le r_i \} \subseteq \spa\{\psi(e) : \head(e) = t\}$$

For convenience, let $d_1 = \sum_{i \in \dem(t)} r_i$, and $d_2 = \maxflow(S^* \setminus \bigcup_{i \in \dem(t)}\{s^*_i\}, t).$
And let $U_1$ be the subspace of $\mathbb{F}^r$ spanned by $b_{i, j}$ with $i \in \dem(t)$, $j \in [r_i]$, and let $U_2$ be the
subspace spanned by $b_{i, j}$ with $i \not\in \dem(t)$, $j \in [r_i]$, which is the complement of $U_1$.

Since $\maxflow(S^* \setminus \bigcup_{i \in \dem(t)}\{s^*_i\}, t) = d_2$, by maximum flow minimum cut theorem, there exists an $(S^* \setminus \bigcup_{i \in \dem(t)}\{s^*_i\})$-$t$ cut
$(V_S, V_T)$ with size $d_2$, where we denote by $e_1, \ldots, e_{d_2}$ all the edges in cut $(V_S, V_T)$. By Lemma \ref{lem:fun}, we know that, for every $e$ with $\head(e) = t$,
$$
\psi(e)|_{U_2} \in \spa\{\psi(e_i)|_{U_2} : i = 1, \ldots, d_2\}.
$$
Note that, $\psi(e)|_{U_1} \in U_1.$ Thus,
$$
\psi(e) \in U_1 + \spa\{\psi(e_i)|_{U_2} : i = 1, \ldots, d_2\},
$$
where $\dim(U_1 + \spa\{\psi(e_i)|_{U_2} : i = 1, \ldots, d_2\}) = d_1 + d_2$.

On the other hand, by Lemma \ref{lem:rnd} and the condition that $\maxflow(S^*,t)=d_1 + d_2$, we claim
$$
\rank\{\psi(e) : \head(e) = t\} \ge d_1 + d_2
$$
holds with probability $\to 1$ ($|\mathbb{F}| \to +\infty$), which implies $U_1 + \spa\{\psi(e_i)|_{U_2} : i = 1, \ldots, d_2\} = \spa\{\psi(e) : \head(e) = t\}$, i.e., the random network works.

\vspace{0.3cm}
For the ``only'' if direction, let's assume for contradiction that $d_1 + d_2 \not= \maxflow(S^*, t)$ and sink $t$ can decode all symbols $b_{i, j}$ with $i \in \dem(t)$, $1 \le j \le r_i$. Noting that $d_1 + d_2 \ge \maxflow(S^*, t)$ always holds by the definition of flow, we may assume $d_1 + d_2 > \maxflow(S^*, t)$.

By Lemma \ref{lem:fun} and Lemma \ref{lem:rnd}, we have
\begin{equation}
\rank\{\psi(e)|_{U_1} : \head(e) = t\} = d_1,
\end{equation}
and
\begin{equation}
\label{equ:U2}
\rank\{\psi(e)|_{U_2} : \head(e) = t\} = d_2
\end{equation}
%By Lemma \ref{lem:rnd} and $\maxflow(S^*, t) = d_1 + d_2$, we know
%\begin{equation}
%\rank\{\psi(e) : \head(e) = t\} = d_1 + d_2
%\end{equation}
%holds with probability $\to 1$. Since $U_1$ and $U_2$ are complementary,
%\begin{eqnarray*}
%& & \rank\{\psi(e) : \head(e) = t\}\\
%& \le & \rank\{\psi(e)|_{U_1} : \head(e) = t\}  + \rank\{\psi(e)|_{U_2} : \head(e) = t\},
%\end{eqnarray*}
%which implies
%\begin{equation}
%\label{equ:U2}
%\rank\{\psi(e)|_{U_1} : \head(e) = t\} = d_1, \rank\{\psi(e)|_{U_2} : \head(e) = t\} = d_2
%\end{equation}
hold with probability $\to 1$ when $|\mathbb{F}| \to +\infty$.
Using Lemma \ref{lem:rnd} again by taking $S_1 = S$, we have
\begin{equation}
\label{equ:rankall}
\rank\{\psi(e)|_{U_2} : \head(e) = t\} \le \maxflow(S^*, t) < d_1 + d_2.
\end{equation}

If $t$ can decode all symbols from sources $\dem(t)$, i.e.,
$$\{ b_{i, j} : i \in \dem(t), 1 \le j \le r_i\} \subseteq \spa\{\psi(e) : \head(e) = t\},$$
then, by \eqref{equ:U2}, we have
$$\{ \psi(e_i)|_{U_2} : 1 \le i \le d_2 \} \subseteq \spa\{\psi(e) : \head(e) = t\},$$
which implies $\rank \{\psi(e) : \head(e) = t\} \ge d_1 + d_2$, contradicted with \eqref{equ:rankall}.
\end{proof}

By the above theorem, it's easy to verify the achievable rate of random linear network coding is monotone,
i.e., the achievability of $(r_1, \ldots, r_m)$ implies the achievability of $(r'_1, \ldots, r'_m)$ as long as
$r'_i \le r_i$ for all $i = 1, 2, \ldots, m$.

For further research, we feel the following question (not well formulated) is interesting.

\vspace{0.2cm}\textbf{Open Question.} For general multi-source multi-sink network coding, is there any computationally feasible solution which
achieves better performance than random network coding?

\section{Acknowledgement}
The author would like to thank his advisor Alexander Razborov for helpful discussions and comments, and his
friends Chen Yuan and Xunrui Yin for helpful discussions.


\begin{thebibliography}{99}
\bibitem{Hal74}
P.~Halmos, ``Finite-Dimensional Vector Spaces,'' Springer-Verlag, 1974.

\bibitem{ACLY00}
R.~Ahlswede, N.~Cai, S.-Y.~Li, and R.~Yeung, ``Network information flow,'' IEEE Trans. Inf. Theory,
vol. 46, no. 4, pp. 1204--1216, Jul. 2000.

\bibitem{LYC03}
S.-Y.~Li, R.~Yeung, and N.~Cai, ``Linear network coding,'' IEEE Trans. Inf. Theory,
vol. 49, no. 2, pp. 371--381, Feb. 2003.

\bibitem{SET03}
P.~Sander, S.~Egner, and L.~Tolhuizen, ``Polynomial time algorithms for network information flow,'' in Proc. 15th ACM Symp. Parallel Algorrithms and Archiectures, San Diego, CA, Jun. 2003, pp. 286--294.

\bibitem{KM03}
R.~Koetter and M.~M$\acute{\text{e}}$dard, ``An algebraic approach to network coding,'' IEEE/ACM Trans. on Networking, vol. 11,
no. 5, Oct. 2003.

\bibitem{DFZ05}
R.~Dougherty, C.~Freiling, and K.~Zeger, ``Insufficiency of linear coding in network information flow,'' IEEE Trans. Inf. Theory,
vol. 51, no. 8, pp. 2745--2759, Aug. 2005.

\bibitem{FS06}
M.~Feder, D.~Ron, and A.~Tavory, ``Bounds on Linear Codes for Network Multicast,'' Electronic Colloquium on Computational Complexity, Report No. 33, 2003.

\bibitem{HMKK06}
T.~Ho, M.~M$\acute{\text{e}}$dard, R.~Koetter, D.~R.~Karger, M.~Effros, J.~Shi, and B.~Leong, ``A random linear network
coding approach to multicast,'' IEEE Trans. Inf. Theory,
vol. 52, no. 10, pp. 4413--4430, Oct. 2006.

\bibitem{YYZ07}
X.~Yan, R.~W.~Yeung, and Z.~Zhang, ``The capacity region for multi-source multi-sink network coding,'' in Proc. IEEE Int.
Symp. Inf. Theory, Nice, France, Jun. 2007, pp. 116--120.

\bibitem{AS08}
N.~Alon and J.~H.~Spencer, ``The Probabilistic Method,'' 3rd edition, John Wiley \& Sons, 2008.

\bibitem{LS11}
M.~Langberg, A.~Sprintson, ``On the hardness of approximating the network coding capacity,'' IEEE Trans. Inf. Theory,
vol. 57, no. 2, pp. 1008--1014, Feb. 2011.

\end{thebibliography}
\end{document}